\documentclass[a4paper,oneside,11pt]{article}
\usepackage{a4wide, dsfont}
\usepackage{graphicx}
\usepackage[latin1]{inputenc}
\usepackage{amsmath,amssymb,amsthm}
\usepackage{url}
\usepackage{ifthen}
\usepackage{natbib}
\usepackage{color}
\usepackage{subfigure}
\usepackage{epstopdf}
\usepackage{caption}
\usepackage{arydshln}
\theoremstyle{plain}
\newtheorem{theorem}{Theorem}%
\newtheorem{prop}[theorem]{Proposition}%
\newtheorem{lemma}[theorem]{Lemma}%
\theoremstyle{definition}
\newtheorem*{example*}{Example}%
\newtheorem*{remark*}{Remark}%
\renewcommand\theta{\vartheta}

\newcommand{\R}{\ensuremath{\mathbb{R}}}%
\newcommand{\B}{\ensuremath{\mathcal{B}}}%
\newcommand{\cF}{\ensuremath{\mathcal{F}}}%
\newcommand{\cX}{\ensuremath{\mathcal{X}}}%
\newcommand{\cp}{\ensuremath{\alpha}}%
\newcommand{\cq}{\ensuremath{\beta}}%
\newcommand{\cP}{\ensuremath{\mathcal{P}}}%
\newcommand{\cW}{\ensuremath{\mathcal{W}}}%
\newcommand{\by}{\ensuremath{{\bf x}}}%
\newcommand{\s}{\ensuremath{{\bf s}}}%
\newcommand{\ind}{{\boldsymbol 1}}

\renewcommand\theta{\vartheta}

\newboolean{comment}
\newcommand{\Kommentar}[1]{%
        \ifthenelse{\boolean{comment}}{#1}{}
        }
\setboolean{comment}{true}

\begin{document}
\bibliographystyle{natbib}
%
%
{\title{Weighted scoring rules and hypothesis testing}
\author{Hajo Holzmann \\
        \small{Fachbereich Mathematik und Informatik}  \\
        \small{Philipps-Universität Marburg} \\
        \small{holzmann@mathematik.uni-marburg.de}
    \and
        Bernhard Klar \\
        \small{Institut für Stochastik}  \\
        \small{Karlsruher Institut für Technologie (KIT)} \\
        \small{bernhard.klar@kit.edu} }
\maketitle
%

%
 }
\begin{abstract}
We discuss weighted scoring rules for forecast evaluation and their connection to hypothesis testing. First, a general construction principle  for strictly locally proper weighted scoring rules based on conditional densities and scoring rules for probability forecasts is proposed. We show how likelihood-based weighted scoring rules from the literature fit into this framework, and also introduce a weighted version of the Hyvärinen score, which is a local scoring rule in the sense that it only depends on the forecast density and its derivatives at the observation, and does not require evaluation of integrals. 	
Further, we discuss the relation to hypothesis testing. Using a weighted scoring rule introduces a censoring mechanism, in which the form of the density is irrelevant outside the region of interest. For the resulting testing problem with composite null - and alternative hypotheses, we construct optimal tests,  and identify the associated weighted scoring rule. As a practical consequence, using a weighted scoring rule allows to decide in favor of a forecast which is superior to a competing forecast on a region of interest, even though it may be inferior outside this region. A simulation study and an application to financial time-series data illustrate these findings.
\end{abstract}
%

%
%
%
\section{Introduction}
Generating and evaluating forecasts is a central task in many scientific disciplines such as makroeconomics, finance or climate and weather research.
While point forecasts for parameters such as the mean or a quantile are more frequently issued \citep{gneiting}, probabilistic forecasts for the whole predictive distribution are most informative and generally preferable \citep{dawid}.
Comparisons of distinct forecasts should be based on proper scoring rules \citep{gneitingraft}, which encourage the forecaster to be honest and make careful assessments according to her true believes.

Sometimes, interest focuses only on certain regions of the whole potential range of the outcome. As a consequence, forecasts should mainly or even exclusively be ranked according to their performance within these regions, while their performance outside is only of minor or no interest. One may think of inflation or GDP growth below a certain threshold, losses of a financial position above a certain threshold, or temperatures or the amount of rain below a certain threshold. Thus, regions of interest for real-valued quantities are often one-sided, and part of the tail of the distribution.

To accommodate such situations, \citet{giacomini} introduced a weighted version of the logarithmic score $ S(p,y) = - w(y)\, \log p(y)$, where $w(y)$ is the weight function such as $w(y) = 1_A(y)$ for some set $A$, and $p(y)$ is the forecast density. However, as observed in \citet{diks} and \citet{gneitingranj}, this is not a proper scoring rule. Indeed, it favors forecasts which put more mass into the region of interest than does the true conditional distribution.

As a remedy, \citet{diks} introduced the conditional and the censored likelihood rules, which depend on weight functions but are proper scoring rules, while
\citet{gneitingranj} introduced weighted versions of the continuous ranked probability score (CRPS). \citet{pelenis} introduces and discusses relevant theoretical properties of weighted scoring rules such as preference preservation, and proposes a penalized likelihood scoring rule as well as an incremental CRPS which satisfy these requirements.

In this paper we propose a general construction principle  for strictly locally proper weighted scoring rules based on conditional densities and on scoring rules for probability forecasts. We show how the likelihood-based weighted scoring rules from \citet{diks} and \citet{pelenis} fit into this framework and  how they are related. Further, we introduce a weighted version of the Hyvärinen score, which is a local scoring rule in the sense of \citet{Ehmgneiting} and \citet{parry} in
that it only depends on the forecast density and its derivatives at the observation, and does not require evaluation of integrals. 	

Our main contribution is to relate the use of weighted scoring rules to hypothesis testing.
Following the simulation setting in \citet{diks} and \citet{lerch}, we cast the Diebold-Mariano test \citep{dm, giacominowhite} into a stylized framework in which two densities are compared based on independent observations.
\citet{lerch} argue that if one density is the true data-generating density, the optimal test is given by the Neyman-Pearson test which corresponds to testing based on the ordinary, non-weighted logarithmic score. Further, in simulations they do not find systematic improvement when comparing two misspecified densities by weighted scoring rules, thus casting doubt on the general usefulness even of weighted scoring rules which are proper.

We argue that when using weight functions and hence focusing on a region of interest $A$, the aim is to ignore possible problems or advantages of the density forecast outside of $A$. Thus, even if a density forecast performs poorly on $A^c$ but well on $A$, it is useful to us if we only focus on the region $A$, indeed as useful as another density forecast which performs well overall.
Hence, using a weighted scoring rule introduces a censoring mechanism, in which the form of the density is irrelevant outside the region of interest. For the resulting testing problem with composite null - and alternative hypotheses, we construct optimal tests, for which the test statistic is the difference in values of the censored likelihood rule of \citet{diks}.

The paper is organised as follows. In Section \ref{sec:weightedscore} we present our theoretical results on the construction of weighted scoring rules. Section \ref{sec:testing} discusses the relation of using weighted scoring rules to hypothesis testing. Section \ref{sec:sims} illustrates the findings in a simulation study, while Section \ref{sec:app} gives an empirical application to financial time series data. Section \ref{sec:discuss} concludes. Proofs are deferred to Section \ref{sec:proofs}.


\section{Weighted scoring rules}\label{sec:weightedscore}

We shall work over an abstract measurable space $(\cX, \cF)$, but essentially think of $(\R, \B)$ or the multivariate situation $(\R^d, \B^d)$. Let $\mu$ be a $\sigma$-finite measure on $(\cX, \cF)$, e.g.~Lebesgue measure, and consider a set $\cP$ of probability densities w.r.t.~$\mu$ on $(\cX, \cF)$.

A {\sl scoring rule} is a map $S :   \cP \times \cX \to \R$, for which for every $p \in \cP$ the map $x \mapsto S(p,x)$ is quasi-integrable for every $q \in \cP$, and for which
\[ S(p,q) = \int_\cX\, S(p,x)\, q(x)\, d\mu(x) > -\infty \quad \text{and} \quad  S(q,q) \in \R \]
for every $p,q \in \cP$. A scoring rule is called {\sl proper} if
\begin{equation}\label{eq:proper}
 S(p,q) \geq S(q,q),\quad q,p \in \cP,
\end{equation}
and it is called {\sl strictly proper} if it is proper and if there is equality in (\ref{eq:proper}) if and only if $p=q$ $\mu$-almost everywhere.
Note the normalization: $S(p,x)$ denotes the loss, and we aim to minimize the expected loss.

We shall consider scoring rules which depend on weight functions, e.g.~measurable functions $w : \cX \to [0,1]$, and use notation and terminology which is closely related to that of \citet{pelenis}. Write $S(p,x;w)$, so that a {\sl weighted scoring rule} is a map $S: \cP \times \cX \times \cW \to \R$ such that $S(\cdot, \cdot;w)$ is a scoring rule for each $w \in \cW$, where $\cW$ is a set of weight functions.
The weighted scoring rule is called {\sl localizing} if
\begin{equation}\label{eq:locallyproper}
 S(h,x;w) = S(p,x;w)  \text{ for }\mu-\text{a.e. }x\in \cX \quad \text{if} \quad p=h \text{ on } \{w>0\} \ \mu-\text{a.e.},\ p,h \in \cP.
\end{equation}
Thus, a localizing weighted scoring rule only depends on the values of the forecast densities on the set $\{w>0\}$ for each $w\in \cW$. Integrating (\ref{eq:locallyproper}) we find that
\[
 S(h,q;w) = S(p,q;w)  \quad \text{if} \quad p=h \text{ on } \{w>0\} \ \mu-\text{a.e.},\ p,q,h \in \cP.
\]%
In particular, if the localizing weighted scoring rule $S$ is also proper, i.e.~$S(\cdot, \cdot ;w)$ is proper for each $w \in \cW$, then
\begin{equation}\label{eq:locallizing}
  S(h,q;w) \geq S(q,q;w) = S(p,q;w)  \quad \text{if} \quad p=q \text{ on } \{w>0\} \ \mu-\text{a.e.},\ p,q,h \in \cP.
\end{equation}
Further, a localizing proper weighted scoring rule is called {\sl strictly locally proper} if $S(p,q;w) = S(q,q;w)$ implies $p=q$ on $\{ w>0\}$ $\mu$-a.e., $p,q \in \cP$, and it is called {\sl proportionally locally proper} if $S(p,q;w) = S(q,q;w)$ if and only if $p= c\, q$ on $\{ w>0\}$ $\mu$-a.e., for some constant $c>0$ which depends on $p,q \in \cP$.

We shall assume that the class of densities $\cP$ and the class of weight functions $w \in \cW$ are such that
\[ \int_\cX p(x)\, w(x)\, d\mu(x) = : \int pw  >0.\]
%


For $p \in \cP$, $w \in \cW$ we let
\[ p_w(x) = \frac{w(x) \, p(x)}{\int\, w\, p}\]
denote the renormalized density of $p$ w.r.t.~$w$.
For formulating the next result, let $\tilde \cP$ be another class of densities such that $p_w \in \tilde \cP$ for every $w \in \cW$, $p \in \cP$.
We show how to construct proportionally locally proper weighted scoring rules from strictly proper scoring rules. \citet{gneiting}, Theorem 5, has a version of this result for scoring function for evaluating forecasts of certain parameters.
\begin{theorem}\label{lem:diks}
Let $S : \tilde \cP \times \cX \to \R$ be a proper scoring rule. Then
\[ \hat S: \cP \times \cX \times \cW \to \R,\quad \hat S(p,x;w) = w(x)\, S(p_w,x)\]
is a localizing proper weighted scoring rule. Further, if $S$ is strictly proper, then $\hat S$ is proportionally locally proper.
\end{theorem}
\begin{example*}
If applied to the logarithmic score $S_l(p,x) = - \log p(x)$, the theorem yields
\begin{align}\label{eq:clscoringrule}
\begin{split}
\hat S_l(p,x;w) & = - w(x)\, \log p_w(x)\\
 & = - w(x)\, \log p(x) + w(x)\, \log\big(\int p \, w \big) - w(x)\, \log w(x)\\
	& =  S^{CL}(p,x;w) - w(x)\, \log w(x),
\end{split}
\end{align}
thus the result is the conditional likelihood rule suggested by \citet{diks} up to a normalizing term which does not depend on the forecast density $p$. Here, we set $0\, \log(0) = 0\, \log(\infty) = 0$. \hfill $\diamond$
\end{example*}
It is remarkable that even though evaluation of the conditional likelihood rule $S^{CL}$ requires evaluation of the integral $\int p\, w$, which in case of $w = 1_A$ amounts to the probability $P(A)$ under $p$, this scoring rule is only proportionally locally proper and thus insensitive to this probability.
\begin{example*}
Localizing proper weighted scoring rules should not be confused with {\sl proper local scoring rules}, as investigated by \citet{Ehmgneiting} and \citet{parry}. As in these papers, consider first the real-valued situation. The dominating measure is the Lebesgue measure, and a scoring rule is local of order $k$ if $S(p,x)$ only depends on $p$ through the first $k$ derivatives $p(x), p'(x),\ldots, p^{(k)}(x)$ of $p$ at $x$.  \citet{parry} show existence of proper local scoring rules for any given even order, while proper scoring rules of odd order do not exist. \citet{Ehmgneiting} characterize proper local scoring rules of order two.

In order to apply Theorem \ref{lem:diks} to higher-order local scoring rules, we require the weight function to be sufficiently smooth as well. This excludes indicator functions, but allows for smooth approximations of indicators. Now, proper, higher-order local scoring rules can be computed without normalizing the density $p$, and thus we expect that the factor $\int p w$ in $\widehat S(p,x;w)$ in Theorem \ref{lem:diks} cancels out, leaving us again with a proper local scoring rule of the same order.

Let us illustrate this for the local scoring rule introduced by \citet{hyvaerinen},
\[ S(p,x) = 2 \, \frac{p''(x)}{p(x)}\, - \Big(\frac{p'(x)}{p(x)}\Big)^2\,.\]
Under conditions on the class of densities stated in \citet{Ehmgneiting} or \citet{hyvaerinen}, it is strictly proper. When applying Theorem \ref{lem:diks}, after canceling terms which only depend on $w$ but not on the forecast density $p$, and dividing by two, we get
\begin{equation}\label{eq:hyvaerinen}
\widehat S(p,x;w) = 2 \, \frac{p''(x)}{p(x)}\, w(x) - \Big(\frac{p'(x)}{p(x)}\Big)^2\, w(x) + 2\, \frac{p'(x)}{p(x)}\, w'(x),
\end{equation}
for which
\begin{align*}
\widehat S(p,q;w)-\widehat S(q,q;w) = \int\, \Big[\frac{p'}{p}- \frac{q'}{q} \Big]^2\, q\, w.
\end{align*}
The rule $\widehat S$ is also a local scoring rule of order two, and belongs to the class characterized by \citet{Ehmgneiting}. Indeed, in their Theorem 3.2, if we put $c=0$ and $K_0(x,y_1) = - w(x)\, y_1^2$, then
\[ s(x,y_0,y_1,y_2) = w(x)\, (y_1^2 + 2 y_2) + 2 w'(x) y_1.\]
Observing
\[ y_1 = \frac{p'}{p},\qquad y_2 = \frac{p''}{p} - \Big(\frac{p'}{p}\Big)^2\]
gives the weighted Hyvärinen score (\ref{eq:hyvaerinen}). Note that, as a locally proper weighted scoring rule, $\widehat S$ is also {\sl preference preserving} in the sense introduced by \citet{pelenis} since it depends linearly on the weight function $w$.

We can also obtain a multivariate version by applying Theorem \ref{lem:diks} to the multivariate score of \citet{hyvaerinen}. For $(\cX, \cF) = (\R^d, \B^d)$, using the notation
\[ \partial_i \, p(\by) : = \frac{\partial}{\partial\, y_i} p(\by), \qquad  \partial_i^2 \, p(\by) : = \Big(\frac{\partial}{\partial\, y_i}\Big)^2 p(\by), \quad i=1, \ldots, d,\]
we obtain
\[ \widehat S(p,\by;w) = \sum_{i=1}^d\, \Big[2 \, \frac{\partial_i^2 \, p(\by)}{p(\by)}\, w(\by) - \Big(\frac{\partial_i \, p(\by)}{p(\by)}\Big)^2\, w(\by) + 2\, \frac{\partial_i^2 \, p(\by)}{p(\by)}\, \partial_i w(\by)\Big]\]
after canceling terms not depending on $p$ and dividing by two, for which we have that
\begin{align*}
\widehat S(p,q;w) - \widehat S(q,q;w) = \int_{\R^d}\, \Big[ \sum_{i=1}^d \Big(\frac{\partial_i \, p(\by)}{p(\by)} - \frac{\partial_i \, q(\by)}{q(\by)} \Big)^2 \Big]\, q(\by)\, w(\by)\, d \by.
\end{align*}
The fact that $\widehat S(p,\by;w)$ does not depend on the normalization constant $\int p w$ may be useful in high-dimensional applications, where computation of $\int p w$ can be difficult. However, as a consequence,  $\widehat S(p,\by;w)$ is only proportionally locally proper and not strictly locally proper. \hfill $\diamond$
\end{example*}

A proportionally  locally proper weighted scoring rule can be turned into a strictly locally proper weighted scoring rule by adding a second weighted scoring rule which takes into account the normalization constant $\int pw$.

\begin{theorem}\label{th:basicresult}
Let $\s(\cp,z)$ be a strictly proper scoring rule for the success probability $\cp \in (0,1)$ of a binary outcome variable $z \in \{0,1\}$. Then
\[ S_\s(p,x;w) = w(x)\, \s\big(\int pw, 1 \big) + \big(1-w(x)\big)\, \s\big(\int pw, 0 \big)\]
is a localizing proper weighted scoring rule for the density forecast $p$. Further, if $S(p,x;w)$ is a  proportionally  locally proper weighted scoring rule, then
\[ \widehat S(p,x;w) = S_\s(p,x;w) + S(p,x;w)\]
is strictly locally proper.
\end{theorem}
By Theorem \ref{th:basicresult}, there are various ways to turn a proportionally  locally proper weighted scoring rule such as the conditional likelihood rule  $S^{CL}$ into a strictly locally proper weighted scoring rule.
Let us illustrate the choices used in the literature to modify $S^{CL}$.
\begin{example*}
The scoring rule for a binary outcome defined by
\begin{equation}\label{eq:binaryscore}
 \bar \s(\cp,z) = - z \, \big(\log \cp + 1 \big) + \cp, \quad \alpha \in (0,1),
\end{equation}
is strictly proper. Indeed, letting
\[ \bar \s(\cp,\cq) = \cq \, \bar \s(\cp,1) + (1-\cq) \, \bar \s\big(\cp,0\big), \quad \cq \in (0,1),\]
we have that
\[ \bar \s(\cp,\cq) - \bar \s(\cq,\cq) = \cq\, \big(\frac{\cp}{\cq} - 1 - \log(\cp/\cq)  \big) \geq 0,\]
since $\log x \leq x-1$, with equality if and only if $x=1$, that is $\cp=\cq$. Moreover,
\[ S_{\bar \s}(p,y;w) =  -  w(y) \big( \log \int wp \big) - w(y) + \int w \, p,\]
and a simple computation shows that
\begin{align*}
S_{\bar \s}(p,x;w) + S^{Cl}(p,x;w)  = S^{PWL}(p,x;w)
\end{align*}
where \begin{align*}
S^{PWL}(p,x;w) = - w(x)\, \log p(x) - w(x) + \int p \, w,
\end{align*}
the penalized weighted likelihood rule by \citet{pelenis}.\hfill $\diamond$
\end{example*}
\begin{example*}
For the logarithmic scoring rule $ \s_l(\cp,z) = - z \log \cp - (1-z)\, \log (1-\cp)$ for a binary outcome we have that
\[ \s_l(\cp,z) = \bar \s(\cp,z) + \bar \s(1-\cp,1-z),\]
where $\bar s(\cp,z)$ is defined in (\ref{eq:binaryscore}),
and one obtains the censored likelihood rule of \citet{diks},
\begin{align}\label{eq:censlike}
S^{CSL}(p,x;w) & = - w(x)\, \log p(x) - \big(1-w(x)\big)\, \log\big(1-\int w p \big)\\
& = S^{CL}(p,x;w) + S_{\s_l}(p,x;w)\notag\\
& = S^{PWL}(p,x;w) + S_{\bar\s}(p,x;1-w).\notag
\end{align}
Thus, the penalized likelihood rule by \citet{pelenis} is ``between'' the conditional and the censored likelihood rules. It is already strictly locally proper, but places less weight on the normalization constant $\int p w$ than the censored likelihood rule. We shall further compare their behaviour in the simulation section. \hfill $\diamond$
\end{example*}
We conclude this section by discussing weighted versions of the continuous ranked probability score (CRPS).
\begin{example*}
\citet{gneitingranj} introduced weighted versions of the continuous ranked probability score (CRPS) for real-valued quantities, both for its representations in terms of probability forecasts and of quantile forecasts. Let
\[ F_p(x) = \int_{- \infty}^x\, p(t)\, dt\]
denote the distribution function associated with the Lebesgue density $p$. If $p$ has a finite first moment, the CRPS can be represented in terms of probability forecasts as
\[ CRPS(p,x) = \int_{- \infty}^{\infty} \big(F_p(z) - 1\{x \leq z\} \big)^2\, dz,\]
and the weighted version in \citet{gneitingranj} in terms of probability forecasts is
\begin{equation}\label{eq:weightedcrps}
 S(p,x;w) = \int_{- \infty}^{\infty} \big(F_p(z) - 1\{x \leq z\} \big)^2\, w(z)\, dz.
\end{equation}
The weighted version of the CRPS remains proper for every $w$. \citet{pelenis} shows that it is not a localizing weighted scoring rule if the class of weight functions contains indicators of compact intervals $w(x) = 1_{[a,b]}(x)$, $a <b$. However, we have the following result.
\begin{prop}\label{prop:wcrps}
For the class of one-sided weight functions
\[ \cW_{os} = \{ w(x) = 1_{[r, \infty)}(x),\ r \in \R\},\]
the weighted CRPS in (\ref{eq:weightedcrps}) is a localizing and strictly locally proper scoring rule.
\end{prop}
\citet{pelenis} proposes a variant of the weighted CRPS called  incremental CRPS. When well-defined, it is localizing and actually strictly locally proper. However, the defining integral is infinite for one-sided weight functions $1_{[r, \infty)}$.
Therefore, in this paper we focus on the weighted CRPS from \citet{gneitingranj} as defined in (\ref{eq:weightedcrps}). \\

In terms of quantile forecasts, the CRPS is given by
\[ CRPS(p,x) = \int_0^1 QS_\alpha\big(F_p^{-1}(\alpha),x\big)\, \mathrm{d}\alpha,\quad QS_\alpha(q,x) = 2\, \big(1_{x <q} - \alpha \big)(q-x),\]
and where $F_p^{-1}$ is the quantile function of $F_p$. For a weight function $v: (0,1) \to [0, 1]$, \citet{gneitingranj} define the quantile-weighted version of the CRPS as
\[ QCRPS(p,x;v) = \int_0^1 QS_\alpha\big(F_p^{-1}(\alpha),x\big)\, v(\alpha)\, \mathrm{d}\alpha.\]
This is not a weighted scoring rule and hence cannot be localizing in the sense of this paper, since the weight function is not defined on sample space $\R$ but rather on $(0,1)$. However, it satisfies another interesting property. Assume that the distribution functions are strictly increasing on their support, so that quantiles are unique and the quantile curve is continuous. If we choose $v(\alpha) = 1_{[r,1)}(\alpha)$, $r \in (0,1)$, then $QCRPS(p,q;v) = QCRPS(q,q;v)$ if and only if $F_p^{-1}(\alpha) = F_q^{-1}(\alpha)$ for all $\alpha \in [r, 1)$. Equivalently, $F_p^{-1}(r) = F_q^{-1}(r)$, and $p=q$ Lebesgue-a.e.~on $\big[F_p^{-1}(r), \infty)$. Thus, the quantile-weighted CRPS evaluates the forecast density $p$ on a forecast-dependent region of interest. \hfill $\diamond$\\

\end{example*}

\section{Relation to hypothesis testing}\label{sec:testing}

Following \citet{lerch} we interprete weighted scoring rules from the perspective of hypothesis testing.
As in the simulations in \citet{diks}, they cast the Diebold-Mariano test \citep{dm, giacominowhite} into a framework in which two densities are compared based on independent observations.
\citet{lerch} argue that if one density is the true data-generating density, the optimal test is given by the Neyman-Pearson test which corresponds to testing based on the ordinary logarithmic score. Thus, improvement by using weighted scoring rules can only be expected when comparing two misspecified densities. However, in their simulations they find no systematic improvement.
Similarly, when comparing a t-distribution with a normal distribution, neither do the simulations in \citet{diks} make a clear case for the use of weighted scoring rules.

In this section we focus on weight functions of the form $w(x) = 1_A(x)$, and argue that when using this weight function, the aim is to ignore possible problems or advantages of the density forecast on the set $A^c$. Thus, even if a density forecast $p$ performs poorly on $A^c$ but well on $A$, it is useful to us if we only focus on the region $A$, indeed as useful as another density forecast which performs well overall. Further, if the focus is on the region $A$, such a forecast $p$ is to be preferred to a forecast which performs well on $A^c$ but poorly on $A$.
This use of weighted scoring rules is not brought to light in the simulations of \citet{lerch}: In their setting, interest focuses on the right tail but all density forecasts compared are correctly specified in the left tail, and ignoring that region does not result in an increased power.
In \citet{diks}, roughly speaking the deviation from the true density is equal on $A$ and on $A^c$ due to the symmetry of the t- and the normal distributions. Thus focusing on one tail will not  give an improvement.

\medskip

In this section we formalize in a stylized setting with one-sided tests what can be achieved by using weighted scoring rules. The simulations in the next section complement our theoretical analysis. Again we work on an abstract measurable space $(\cX, \cF)$, and consider a region of interest $A \in \cF$, corresponding to the weight function $w(x) = 1_A(x)$. Let $p_0$, $p_1$ be two competing densities w.r.t.~the dominating measure $\mu$, and let $P_0$, $P_1$ denote the associated probability measures. Assume that $0 < P_0(A), P_1(A) < 1$.
Since the property (\ref{eq:locallyproper}) of localizing weighted scoring rules implies that the forecasts are only relevant through their values on $A$, testing using score differences with weight function $w(x) = 1_A(x)$ amounts to testing
\begin{equation}\label{eq:testingproblem}
 H_0: p\, 1_A  = p_0\, 1_A \ \mu-\text{a.e.}\quad \text{vs.} \quad H_1: p\, 1_A  = p_1\, 1_A \ \mu-\text{a.e.}
\end{equation}
Here, we have a composite null hypothesis and a composite alternative hypothesis, and for the formulation of $H_0$ only the restriction of $p_0$ to $A$ is relevant. Note that this includes the probabilities $P_0(A^c) = 1-P_0(A)$ and $P_1(A^c)$, but not the shapes of $p_0$ and $p_1$ on $A^c$.
Hence, for a candidate density $p$, only the shape on $A$ and the probabilities $P(A)$ and hence $P(A^c)$ matter for the test decision.

Consider (\ref{eq:testingproblem}) on the basis of repeated observations, e.g.~on
$ \big(\cX^n,\mu^{\otimes n}\big)$,
and product densities $\prod_{k=1}^n p(x_k)$, where $p$ satisfies the restrictions under $H_0$ or $H_1$. For a weighted scoring rule $S$ let
\begin{equation}\label{eq:scorediff}
 T(x_1, \ldots, x_n) = \sum_{k=1}^n \big(S(p_0,x_k;1_A) - S(p_1,x_k;1_A) \big)
\end{equation}
denote the sum of the score differences. Given $\alpha \in (0,1)$ let $c_\alpha$ denote the $1-\alpha$-quantile of $T$ under $P_0^{\otimes n}$, and consider the test
\begin{equation}\label{eq:thescoretest}
\phi(x_1, \ldots, x_n) = \left\{ \begin{array}{cc} 1,& T(x_1, \ldots, x_n) > c_\alpha,\\
\frac{P_0^{\otimes n}(T \leq c_\alpha) - (1-\alpha)}{P_0^{\otimes n}(T \leq c_\alpha) - P_0^{\otimes n}(T < c_\alpha)},& T(x_1, \ldots, x_n) = c_\alpha,\\
0,& T(x_1, \ldots, x_n) < c_\alpha.\\
\end{array}\right.
\end{equation}
In the next result we relate localizing weighted scoring rules to the testing problem (\ref{eq:testingproblem}) and determine the optimal choice of the scoring rule, which turns out to be the censored likelihood rule (\ref{eq:censlike}) as proposed in \cite{diks}.

\begin{theorem}\label{th:scoretesting}
1. Consider the testing problem (\ref{eq:testingproblem}). Given $\alpha \in (0,1)$, if $S$ is a localizing weighted scoring rule, then the test $\phi$ in (\ref{eq:thescoretest}) for which $T$ in (\ref{eq:scorediff}) is based on $S$, has a constant power $\alpha$ on $H_0$, and a constant power on $H_1$.

2. Given $\alpha \in (0,1)$, the test $\phi_{cen}$ using the censored likelihood rule (\ref{eq:censlike}) is uniformly more powerful then tests of the form (\ref{eq:thescoretest}) based on other localizing weighted scoring rules.
\end{theorem}
Without using weighted and localizing scoring rules, score differences $T$ do not easily lead to tests for the testing problem (\ref{eq:testingproblem}). The critical value for the rejection region would need to be computed as a supremum of $1-\alpha$-quantiles of $T$ under all densities $p$ satisfying $H_0$, which is generally infeasible.

Finally, we show that $\phi_{cen}$ is a minimax test in the class of all level-$\alpha$ tests for (\ref{eq:testingproblem}) in the sense that it maximizes the minimal power. Further, in case of a single observation, $\phi_{cen}$ is even uniformly most powerful in this general class of tests. Define the minimal power of a test $\phi: \cX^n \to [0,1]$ as
\begin{equation}\label{eq:power}
\rho(\phi) =  \inf_{p:\ p\, 1_A  = p_1\, 1_A} \int_\cX \cdots \int_\cX \,  \phi(x_1, \ldots x_n)\, \prod_{k=1}^n p(x_k)\, d \mu(x_1)\ldots d \mu(x_n)
\end{equation}
\begin{theorem}\label{th:optimalitytest}
Given $\alpha \in (0,1)$, let $\phi: \cX \to [0,1]$ be any level-$\alpha$-test for (\ref{eq:testingproblem}), and let $\phi_{cen}$ be the test in (\ref{eq:thescoretest}) of level $\alpha$ when using the censored likelihood rule (\ref{eq:censlike}). Then $\rho(\phi_{cen}) \geq \rho(\phi)$. In case of a single observation $n=1$, the test $\phi_{cen}$ is even uniformly most powerful.
\end{theorem}
For $n>1$ we were not able to clarify whether $\phi_{cen}$ is uniformly most powerful in a larger calls of tests than those given in (\ref{eq:thescoretest}).

\section{Simulation} \label{sec:sims}

In this section, we consider simulation settings similar to those in \cite{diks} and \cite{lerch}.
Suppose that at time $t = 1,\ldots,n$, the observations $y_t$ are independent standard normally distributed. We apply the two-sided Diebold-Mariano test of equal predictive performance, nominal level $\alpha=0.05$, using the variance estimate in display (2.17) of \cite{lerch} with $k=1$. As nonparametric alternative, we apply the one-sided Wilcoxon signed-rank test, nominal level $\alpha=0.025$.

\begin{table}
\centering{
\begin{tabular}{cc|ccccc}
\hline
&scoring rule & proper & strictly & localizing & strictly  & proportionally \\
&                  &            & proper &               & locally proper & locally proper \\  \hline
&CRPS          & yes      & yes      & no & - & - \\
unweighted &LogS           & yes      & yes      & no & - & - \\
&HY              & yes      & yes      & no & - & - \\ \hline
&twCRPS      & yes      & -          & no (yes)  & no (yes)  & no \\
&CSL            & yes      & -          & yes & yes & no \\
weighted &CL              & yes      & -          & yes & no  & yes \\
&PWL           & yes       & -          & yes & yes & no \\
&WH            & yes       & -          & yes & no  & yes \\ \hline
\end{tabular}
\caption{\label{table1} Summary of properties of unweighted and weighted scoring rules. The entry {\it no (yes)} for twCRPS indicates that it is localizing and strictly locally proper for the one-sided weight functions used in the simulations, but not in general.}
}
\end{table}

In the simulations, we use  the logarithmic score (LogS), the continuous ranked probability score (CRPS) and the Hyvärinen score (HY) as typical examples of unweighted scoring rules. As weighted scoring rules, we apply the threshold weighted continuous ranked probability score (twCRPS), the censored likelihood rule (CSL), the conditional likelihood rule (CL), the penalized weighted likelihood rule (PWL) and the weighted Hyvärinen score (WH). Table \ref{table1} gives a summary of the properties of these scoring rules.

Suppose that we are only interested in the forecast quality on a subset of the support of the underlying distribution. For example, interest may center on the positive real numbers or on the right tail of the distribution. Hence, the tests under the twCRPS, CL, CSL, PWL and WH scoring rules use the indicator weight function $w(z) = \ind\{z \ge r\}$ in all simulations.
For the weighted Hyvärinen score, we approximate the weight function $w(z) = \ind\{z \ge r\}$ with the $C^1$-function
\begin{align*}
  \tilde{w}(y) &= 3y^2-2y^3, \quad   \tilde{w}'(y) = 6y(1-y), \quad  y\in (0,1),
\end{align*}
(and $\tilde{w}(y)=\tilde{w}'(y)=0, y\le 0, \, \tilde{w}(y)=1, \tilde{w}'(y)=0, y\ge 1$), shifted and scaled to the interval $(r-\delta,r+\delta)$, with $\delta$ chosen as $0.5$.

We either use the fix sample size $n=100$, or we vary the sample size with the threshold value $r$ in such a way that under the standard normal distribution the expected number, $c = 10$, of observations in the relevant region $[r,\infty)$ remains constant. \\

\vspace{3mm}
{\bf Scenario A1:} Forecast 1: $\Phi$  vs. Forecast 2: $F_{hlt}$.

Here, $\Phi$ denotes the cumulative distribution function (cdf) of the standard normal distribution, and $F_{hlt}$  a piecewise defined distribution with continuous density and heavy left tail, consisting of a scaled $t_4$-distribution on $(-\infty,0]$ and a standard normal distribution on $(0,\infty)$. Clearly, for positive values of $r$, $\Phi$  and $F_{hlt}$ coincide. \\

Figure \ref{fig:scenarioA1-plot} shows the proportion of rejections of the null hypothesis of equal predictive performance in favor of $\Phi$, as a function of the threshold value $r$ in the weight function. For $r<0$, rejections in favor of the standard normal distribution represent true power, but if one is interested in the region $[r,\infty)$ for positive $r$, both forecasts are identical, and neither of them should be rejected. Results for the Diebold-Mariano test and the Wilcoxon signed-rank test for sample size $n=100$ are shown in the left and right panel, respectively. \\

\begin{figure}
\begin{center}
\includegraphics[scale=0.6]{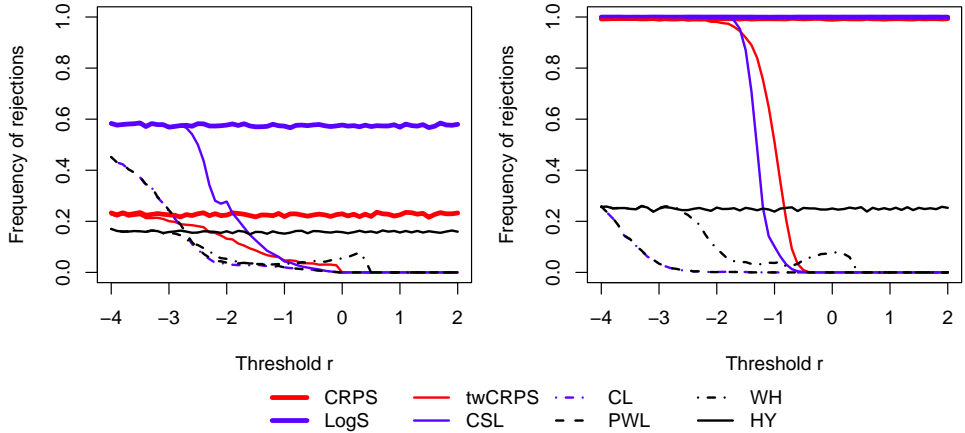}
\caption{Scenario A1. The null hypothesis of equal predictive performance of $\Phi$ and $F_{hlt}$ is tested under a standard normal population. The panels show the frequency of rejections in two-sided Diebold-Mariano test (left) or in one-sided Wilcoxon signed-rank test (right) in favor of $\Phi$ for $n=100$.}
\label{fig:scenarioA1-plot}
\end{center}
\end{figure}

Let us first look at the non-weighted scoring rules. They have rather different rejection frequencies when using the Diebold-Mariano test, with LogS well above CRPS and HY. Using the Wilcoxon signed-rank test, HY has a much lower rejection frequency than the other two non-weighted scoring rules.

Clearly, for large negative values of $r$, the rejection frequencies of CSL, CL and PWL coincide with those of LogS, but the rejection frequencies of CL and PWL, which are nearly identical, decrease faster to zero than for CSL. The WH score reaches zero at 0.5 (instead of 0) due to the specific choice of smooth weight function.

Having in mind that, in this scenario, the goal is to detect differences in predictive performance as long as $r<0$, but not for $r>0$, CSL is the preferable scoring rule using the Diebold-Mariano test, whereas CSL and twCRPS are preferable in case of the nonparametric test.

\vspace{3mm}
{\bf Scenario A2:} Forecast 1: $F_{hlt}$ from Scenario A1 vs. Forecast 2, a piecewise defined distribution $F_{hrt}$ with continuous density and heavy right tail, consisting of a standard normal distribution on $(-\infty,0]$ and a scaled $t_4$-distribution on $(0,\infty)$.

Figure \ref{fig:scenarioA2-plot} shows the proportion of rejections of the null hypothesis of equal predictive performance in favor of $F_{hlt}$, as a function of the threshold value $r$ in the weight function. For $r=-\infty$, both forecasts have the same distance from the (true) standard normal distribution, and neither of them should be rejected in favor of the other. However, for $r>0$, Forecast 1 coincides with $\Phi$, and Forecast 2 should be rejected. Results for the Diebold-Mariano test and the Wilcoxon signed-rank test for sample size $n=100$ are shown in the left and right panel, respectively. \\

\begin{figure}
\begin{center}
\includegraphics[scale=0.6]{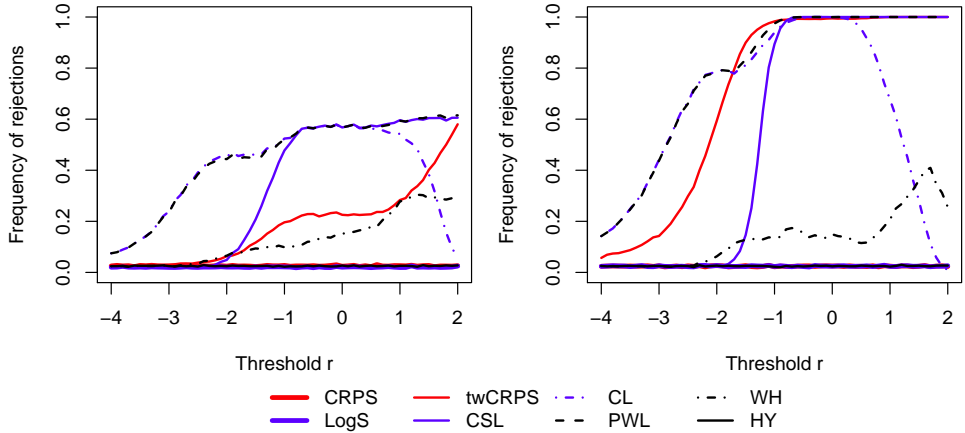}
\caption{Scenario A2. The null hypothesis of equal predictive performance of $F_{hlt}$ and $F_{hrt}$ is tested under a standard normal population. The panels show the frequency of rejections in two-sided Diebold-Mariano test (left) or in one-sided Wilcoxon signed-rank test (right) in favor of $F_{hlt}$ for $n=100$.} \label{fig:scenarioA2-plot}
\end{center}
\end{figure}

As one would expect, the rejection frequencies in favor of $F_{hlt}$ of all non-weighted scoring rules are around 0.025.
CL and PWL show the fastest increase of the rejection frequencies; however, CL decreases to zero for large positive values of $r$. This is due to the fact that the effective sample size, i.e. the number of observations exceeding $r$ becomes very small with increasing threshold.

In this scenario, since a reasonable forecast should not detect differences in predictive performance for large negative values of $r$ ($r<-4$, say), but should detect differences for mildly negative and positive values of $r$, PWL is the preferable scoring rule. It is remarkable that the rejection rate of CSL increases faster than that of twCRPS using the Diebold-Mariano test, but vice versa in case of the Wilcoxon signed-rank test.

\vspace{3mm}
{\bf Scenario B:} Denote the cdf of a normal distribution with mean $\mu$ and standard deviation $\sigma$ by $\Phi_{\mu,\sigma}$. Let
\begin{align*}
G(x) &= \Phi_{0,1/2}(x) \, \Phi(x) + \left(1-\Phi_{0,1/2}(x) \right) F_4(x), \\
H(x) &= \left(1-\Phi_{0,1/2}(x)\right) \Phi(x) + \Phi_{0,1/2}(x) \, F_4(x),
\end{align*}
where $F_4$ denotes the distribution function of the $t$-distribution with 4 degrees of freedom.
Figure \ref{fig:diff-plot-abs} shows the differences of the cdf's $\Phi-G$ and $\Phi-H$ (left) and the relative differences $(\Phi-G)/(1-\Phi(r))$ and $(\Phi-H)/(1-\Phi(r))$ (right).


\begin{figure}
\begin{center}
\includegraphics[scale=0.6]{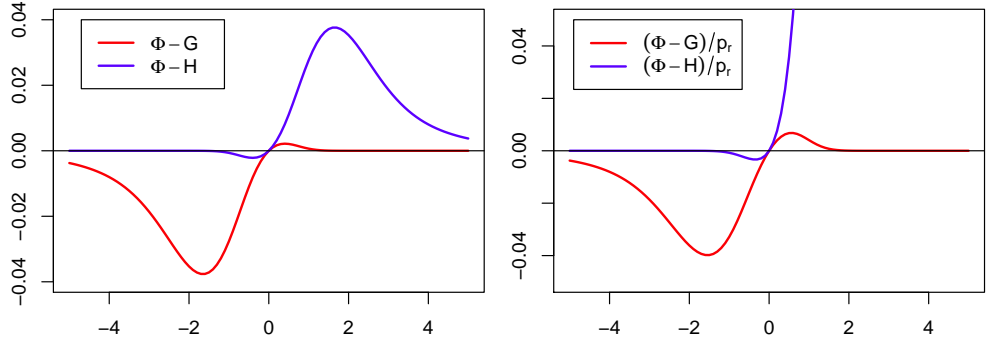}
\caption{Scenario B. Differences of cdf's $\Phi-G$ and $\Phi-H$ (left) and relative differences of cdf's $(\Phi-G)/p_r$ and $(\Phi-H)/p_r$ (right), where $p_r=1-\Phi(r)$.}
\label{fig:diff-plot-abs}
\end{center}
\end{figure}

\vspace{3mm}
In Scenario B, we consider Forecast 1: $G$  vs. Forecast 2: $H$.

Figure \ref{fig:scenarioB-plot1} shows the proportion of rejections of the null hypothesis of equal predictive performance in favor of $G$, as a function of the threshold value $r$ in the weight function for the Diebold-Mariano test (left panels) and the Wilcoxon signed-rank test (right panels). The upper (lower) panels show the results for fixed (varying) sample size.
If one is only interested in the region $[r,\infty)$ for larger values of $r$, forecast $G$ is close to $\Phi$; hence, $H$ should be rejected.

\begin{figure}
\begin{center}
\includegraphics[scale=0.59]{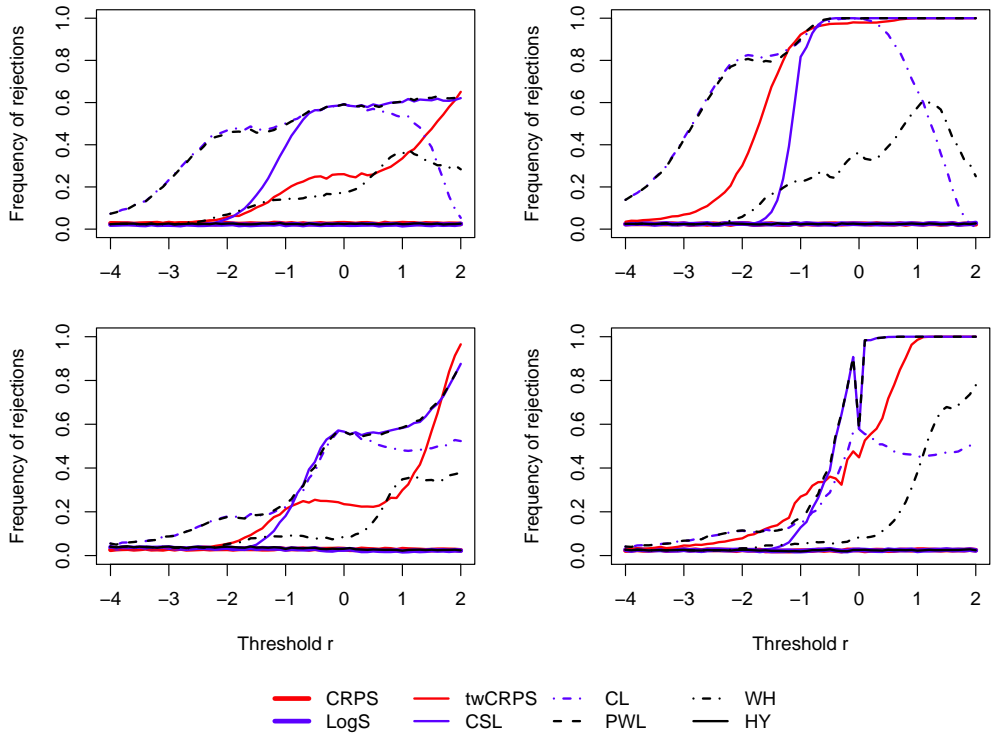}
\caption{Scenario B. The null hypothesis of equal predictive performance of $G$ and $H$ is tested under a standard normal population. The upper panels show the frequency of rejections in two-sided Diebold-Mariano tests (left) or in one-sided Wilcoxon signed-rank test (right) in favor of $G$ for $n=100$. The lower panels show the corresponding plots with $n$ varying in such a way that the expected number of observations in $[r,\infty)$ is 10 under the standard normal distribution.}
\label{fig:scenarioB-plot1}
\end{center}
\end{figure}

Qualitatively, the right panels of Figure \ref{fig:scenarioB-plot1} coincide with the corresponding left panels, but the right panels have much higher rejection frequencies.  The upper panels look very similar to the corresponding panels of Figure \ref{fig:scenarioA2-plot}. As for the lower left panel, the rejection frequencies in favor of $G$ finally increase for all weighted scoring rules if the number of observations in the region of interest remains constant.

Finally, Figure \ref{fig:scenarioB-plot2} shows the proportion of rejections in favor of $H$ for the Diebold-Mariano test. Here, the left (right) panel shows the results for fixed (varying) sample size. Rejection frequencies in favor of $H$ are (approximately) $2.5\%$ for the unweighted scoring rules, they are very low for CL and PWL, and decrease from 0.025 to zero for CSL, twCRPS and WH.

\begin{figure}
\begin{center}
\includegraphics[scale=0.60]{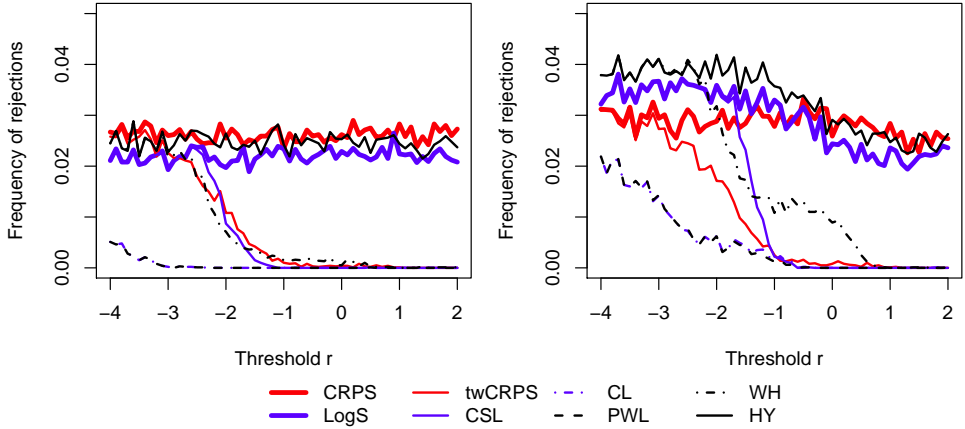}
\caption{Scenario B. The null hypothesis of equal predictive performance of $G$ and $H$ is tested under a standard normal population. The panels show the frequency of rejections in two-sided Diebold-Mariano tests in favor of $H$ for $n=100$ (left) and varying $n$ (right). }
\label{fig:scenarioB-plot2}
\end{center}
\end{figure}



\section{Empirical illustration}\label{sec:app}

We apply the proposed forecasting rules to two time series of daily log returns $y_t = \ln(P_t/P_{t-1})$, where $P_t$ is the closing price on day $t$, adjusted for dividends and splits. We consider S\&P 500 and Deutsche Bank AG log-returns for a sample period running from January 1, 2009 until December 31, 2016, giving a total of 2005 and 2025 observations (source: Yahoo finance). Since Yahoo finance data for Deutsche Bank partially includes holidays, we removed all days with zero trading volume.

We define three forecast methods based on the following GARCH(1,1) model,
\begin{equation}
y_t = µ + \sigma_t z_t, \qquad \sigma_t^2 = \omega + \alpha_1(y_{t-1} - µ)^2 + \beta_1 \sigma_{t-1}^2,
\end{equation}
using normal, t and skew-t distributions for the innovations to account for leptokurtosis and/or skewness. Since a typical finding in empirical applications of GARCH models is that a normal distribution for $z_t$ does not fully account for the kurtosis observed in stock returns, we may expect that the forecast with $t$-distributed innovations gives better density forecasts.

To evaluate the three forecast methods, we use one-step-ahead density forecasts with a rolling window scheme for parameter estimation done by maximum likelihood method using R and the R package rugarch (\citet{ghalanos}, \citet{rcoreteam}). The length of the estimation window is set to be 500 observations, so that the number of out-of-sample observations is equal to 1505 and 1525. For comparing the density forecasts' accuracy we apply the Diebold-Mariano test based on several weighted and unweighted scoring rules. We use the threshold weight function $w(y)=1\{y \leq r\}, r=-1,0,$ and $w(y)=1\{y \geq r\}, r=0,1$, and hence concentrate either on losses or on gains when using the weighted scoring rules. The score difference is computed by subtracting the score of the normal GARCH density forecast from the score of the $t$-GARCH density forecast, so that positive values indicate better predictive ability of the forecast method based on Student-$t$ innovations, and similarly for normal vs. $t$ and $t$ vs. skew-$t$ innovations. The results for the S\&P 500 and Deutsche Bank AG can be found in Tables \ref{table2} and \ref{table3}, respectively.

\begin{table}
\centering{
\begin{tabular}{cc|cccc}   \hline
& $w(z)$ & $1\{z\leq -1\}$ & $1\{z\leq 0\}$ & $1\{z\geq 0\}$ & $1\{z\geq 1\}$ \\
& proportion &  0.12 & 0.45 & 0.56 & 0.14 \\ \hline
& CRPS & 1.26 & 1.26 & 1.26 & 1.26 \\
& LogS & 3.31 & 3.31 & 3.31 & 3.31 \\
normal garch & twCRPS & 0.60 & 0.56 & 1.86 & -0.14 \\
vs. t-garch & CSL & 2.61 & 2.97 & 0.85 & -1.77 \\
&  CL & 2.63 & 3.24 & 1.62 & -3.85 \\
&  PWL & 2.62 & 3.09 & 1.30 & -2.01 \\ \hline
& CRPS & 1.78 & 1.78 & 1.78 & 1.78 \\
&  LogS & 3.53 & 3.53 & 3.53 & 3.53 \\
normal garch &  twCRPS & 1.59 & 1.40 & 1.26 & 0.44 \\
vs. skew-t-garch &  CSL & 2.91 & 3.24 & 0.71 & 0.63 \\
&  CL & 2.56 & 3.54 & 1.51 & -0.65 \\
&  PWL & 2.89 & 3.38 & 1.19 & 0.56 \\ \hline
&  CRPS & 1.22 & 1.22 & 1.22 & 1.22 \\
&  LogS & 1.38 & 1.38 & 1.38 & 1.38 \\
t-garch &  twCRPS & 1.49 & 2.44 & -0.28 & 1.08 \\
vs. skew-t-garch &  CSL & 2.04 & 1.71 & 0.04 & 2.79 \\
&  CL & 0.41 & 1.97 & 0.31 & 3.45 \\
&  PWL & 2.02 & 1.84 & 0.21 & 2.91 \\  \hline
\end{tabular}
\caption{\label{table2} $t$-statistics for Diebold-Mariano test for equal predictive accuracy for S\&P 500.
Positive values indicate superiority of forecasts from the second method,
while negative values indicate superiority of forecasts from the first method.}
}
\end{table}

On the whole, forecasts for the S\&P 500 returns using a $t$ or skew-$t$ GARCH model are superior to a normal GARCH model; using weighted scoring functions, we see that this holds especially for losses, but only to a lesser extent for gains. In particular, the $t$ GARCH model seems to be inferior to the normal GARCH for the threshold weight function $1\{z\geq 1\}$. As can be seen in the lower panel of Table \ref{table2}, results are less clear cut between $t$ and skew-$t$ GARCH density forecasts depending on the weight function, with an overall advantage for the skew-$t$ GARCH model.

For the Deutsche Bank returns, $t$ and skew-$t$ GARCH density forecasts are generally superior to a normal GARCH model for all (weighted and unweighted) scoring functions, but again this holds to a lesser extent for gains, as can be seen in Table \ref{table3}. The lower panel shows that there is no significant overall difference between $t$ and skew-$t$ GARCH density forecasts; however, the skew-$t$ GARCH model is significantly better for predicting losses whereas the $t$ GARCH model is clearly superior for predicting gains.

\begin{table}
\centering{
\begin{tabular}{cc|cccc}   \hline
& $w(z)$ & $1\{z\leq -1\}$ & $1\{z\leq 0\}$ & $1\{z\geq 0\}$ & $1\{z\geq 1\}$ \\
& proportion &  0.30 & 0.50 & 0.50 & 0.32 \\ \hline
&  CRPS & 1.10 & 1.10 & 1.10 & 1.10 \\
&  LogS & 2.30 & 2.30 & 2.30 & 2.30 \\
normal garch &  twCRPS & 0.70 & 0.41 & 0.87 & 1.06 \\
vs. t-garch &  CSL & 1.67 & 1.75 & 1.42 & 1.76 \\
&  CL & 1.72 & 1.80 & 1.52 & 2.00 \\
&  PWL & 1.67 & 1.78 & 1.47 & 1.84 \\ \hline
&  CRPS & 0.84 & 0.84 & 0.84 & 0.84 \\
&  LogS & 2.10 & 2.10 & 2.10 & 2.10 \\
normal garch &  twCRPS & 1.10 & 0.61 & 0.30 & 0.52 \\
vs. skew-t-garch &  CSL & 1.96 & 1.90 & 0.57 & 1.19 \\
&  CL & 1.95 & 2.02 & 0.81 & 1.48 \\
&  PWL & 1.95 & 1.96 & 0.69 & 1.28 \\  \hline
&  CRPS & -0.54 & -0.54 & -0.54 & -0.54 \\
&  LogS & -0.39 & -0.39 & -0.39 & -0.39 \\
t-garch &  twCRPS & 1.38 & 1.06 & -1.67 & -1.44 \\
vs. skew-t-garch &  CSL & 2.38 & 1.45 & -1.93 & -1.33 \\
&  CL & 2.09 & 1.87 & -1.67 & -1.26 \\
&  PWL & 2.30 & 1.70 & -1.83 & -1.32 \\  \hline
\end{tabular}
\caption{\label{table3} $t$-statistics for Diebold-Mariano test for equal predictive accuracy for Deutsche Bank AG.
Positive values indicate superiority of forecasts from the second method,
while negative values indicate superiority of forecasts from the first method.}
}
\end{table}

A visual inspection of the distribution of the residuals sheds some light on the previous results. Since the estimates for $\mu,\omega,\alpha_1$ and $\beta_1$ are very similar for the three models, the resulting empirical distributions of the residuals are visually nearly indistinguishable. Hence, Fig. \ref{fig:example-densities} only shows the empirical density of the residuals under normality assumption together with the theoretical densities of the normal, the fitted $t$-distribution (with shape parameter 8.4) and the fitted skew-$t$-distribution (with shape and skewness parameter 8.5 and 0.94, respectively) for the Deutsche Bank return series.

At first sight, the empirical density looks fairly symmetric, and all three distributions seem to fit the tails quite well, whereas the normal density is not sufficiently peaked in the center. Hence, it is clear that the normal distribution based forecasts will be inferior on the whole. Looking more closely, one actually finds regions in the right tail where the normal distribution fits better than $t$ and skew-$t$; thus, the advantage of the latter diminishes. In the center, the skew-$t$ seems to yield a better fit than the $t$ distribution for values smaller than zero, and vice versa for positive values, which may explain the results of the comparison between $t$ GARCH and skew-$t$ GARCH shown in the lower panel of Table \ref{table3}.

\begin{figure}
\begin{center}
\includegraphics[scale=0.60]{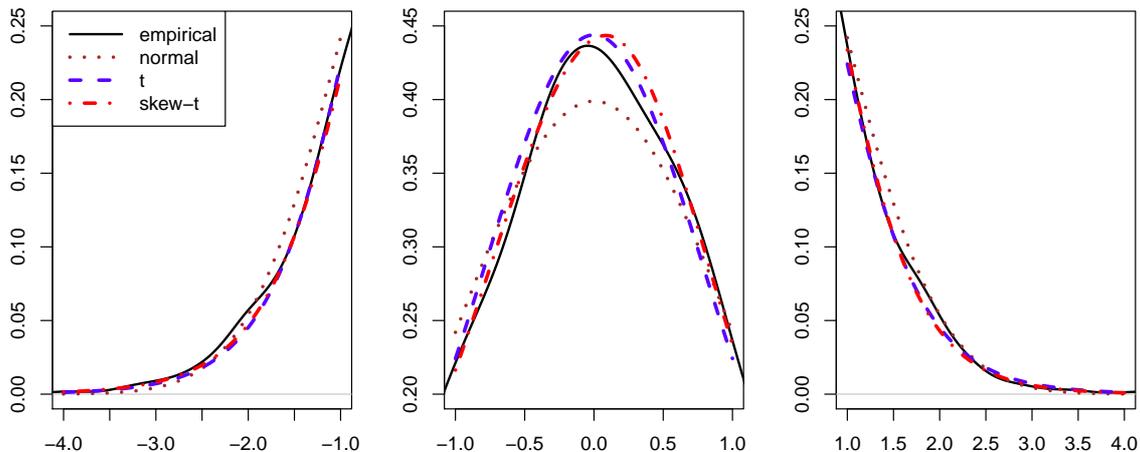}
\caption{Empirical and theoretical density functions of the residuals of a GARCH(1,1)-model fitted to the Deutsche Bank return series. For better visibility, left tail, center and right tail of the distribution are displayed in separate panels. }
\label{fig:example-densities}
\end{center}
\end{figure}

\section{Discussion and conclusions}\label{sec:discuss}

\citet{lerch} discuss the so-called {\sl forecasters dilemma}, in that forecasts are often only evaluated in case that extreme events actually occur. They point out that such a restriction of forecast evaluation to subsets of the available observations has highly unwanted effects, and it discredits even the best possible forecast, that is the true  conditional distribution.

Weighted scoring rules which remain proper are a valid decision-theoretic tool for emphasizing regions of interest, but the analysis and the simulations in \citet{lerch} cast doubts on their usefulness as compared to ordinary, non-weighted scoring rules. We show in a test-theoretical framework that weighted scoring rules indeed can be useful if interest is restricted to a subset $A$ of the potential observational domain. For example, if a forecast performs well on $A$ but poorly on $A^c$, it will still be found useful by localizing proper weighted scoring rules which focus on $A$.

The weighted scoring rule leading to optimal one-sided tests in an i.i.d.~setting is the censored likelihood rule as proposed in \citet{diks}. However, more generally other weighted scoring rules may be used as well, and our simulations show that the penalized likelihood score by \citet{pelenis}, which also has desirable theoretical properties such as being preference preserving, performs well in practice.

\section{Proofs}\label{sec:proofs}

\small

\begin{proof}[{\sl Proof of Theorem \ref{lem:diks}}]
First,
\begin{align*}
\widehat S (p,q;w) = \int q_w(x)\, S(p_w,x)\, d\mu(x) \, \int q w,
\end{align*}
so that $\widehat S$ is also a proper scoring rule for each $w$, and it is localizing since it only depends on $p_w$ and thus on $p$ only on $\{w>0\}$.
Further, if $S$ is strictly proper, then $\widehat S (p,q;w) = \widehat S (q,q;w)$ implies that $p_w = q_w$ $\mu$-a.e. Finally, observe that the densites $p$ and $q$ are proportional on $w>0$ $\mu$-a.e.~if and only if $p_w = q_w$ $\mu$-a.e.
This concludes the proof.
\end{proof}

\begin{proof}[{\sl Proof of Theorem \ref{th:basicresult}}]
The rule $S_\s$ is apparently localizing w.r.t.~$p$ since it depends only on $\int pw$. Further, it is proper since
\begin{equation}\label{eq:propertyprop}
S_\s(p,q;w) - S_\s(q,q;w) = \s\big(\int pw, \int qw \big) - \s\big(\int qw, \int qw \big) \geq 0,
\end{equation}
where we used the notation
\[ \s(\cp,\cq) = \cq \s(\cp,1) + (1-\cq) \s\big(\cp,0\big).\]
Now, as a sum of two locally proper scoring rules the rule $\widehat S$ is also a locally proper scoring rule.
Further, if
\[ \widehat S(p,p;w) = \widehat S(p,q;w),\]
then necessarily $S(p,p;w) =  S(p,q;w)$, which implies that $p = c\,q$ on $w>0$, and also that $S_\s(p,p;w) =  S_\s(p,q;w)$, which by (\ref{eq:propertyprop}) and the fact that $\s$ is strictly proper implies that $\int pw = \int qw$. But since we assume $\int qw \not=0$ and $\int pw \not=0$, we get for the proportionality constant that $c=1$ and hence $p=q$ on $w>0$, so that $\widehat S$ is  strictly locally proper.
\end{proof}
\begin{proof}[{\sl Proof of Proposition \ref{prop:wcrps}}]
If $p=q$ Lebesgue-a.e.~on $[r, \infty)$, then
\[ 1- F_p(z) = \int_z^\infty p(u)\, du = \int_z^\infty q(u)\, du = 1- F_q(z), \quad z \geq r,\]
and hence also $F_p(z) = F_q(z)$, $z \geq r$. Therefore
\[ S\big(p,x;1_{[r, \infty)}\big) = \int_{r}^{\infty} \big(F_p(z) - 1\{x \leq z\} \big)^2\,  dz
\]
is localizing. A computation shows that
\[ S\big(p,q;1_{[r, \infty)}\big) - S\big(p,p;1_{[r, \infty)}\big) = \int_{r}^{\infty} \big(F_p(z) - F_q(z) \big)^2\,  dz,
\]
so that the scoring rule is strictly locally proper.
\end{proof}

For the proofs of the results in Section \ref{sec:testing}, we introduce the following censoring mechanism.
Fix an $\bar a \in A^c$ and let
$ \tilde \cX  = A \cup \{ \bar a\}$, $\tilde \mu (B) = \mu(A \cap B) + \delta_{\bar a}(B)$ for $B \subset \tilde \cX  $ measurable, where $\delta_{\bar a}$ is the Dirac measure at ${\bar a}$, and for $j=0,1$ let
\begin{align*}
\tilde p_j(x) & = \left\{\begin{array}{cc} p_j(x), & \quad x \in A,\\  P_j(A^c),& \quad x = \bar a ; \end{array} \right.
\end{align*}
these are densities w.r.t.~$\tilde \mu$, with associated probability measures $\tilde P_j$. Consider the testing problem on $ \tilde \cX^n$ with simple null - and alternative hypotheses,
\begin{equation}\label{eq:testingproblem2}
 \tilde H_0: \tilde P= \tilde P_0  \quad \text{vs.} \quad H_1: \tilde P= \tilde P_1 ,
\end{equation}
which determine the product densities on $ \tilde \cX^n$. By sufficiency of the order statistics we may assume that tests $\phi(x_1, \ldots x_n)$ on $\cX^n$ and $\tilde \phi(x_1, \ldots x_n)$ on $\tilde \cX^n$ are symmetric in the arguments.
\begin{lemma}\label{lem:mainlemtesting}
1. There is a one-to-one correspondence between tests $ \phi: \cX^n\to [0,1]$ which are constant in arguments varying in $A^c$, given fixed arguments in $A$, and tests $ \tilde \phi: \tilde \cX^n\to [0,1]$, which is given by
\[ \phi(x_1, \ldots, x_n) = \tilde \phi(x_1, \ldots, x_j, \bar a, \ldots, \bar a),\quad \text{ if }x_1, \ldots, x_j \in A,\ x_{j+1}, \ldots, x_n \in A^c,\,\]
$0 \leq j \leq n$, and extended by symmetry. \\
2. Given $\alpha \in (0,1)$, the test $\phi_{cen}$ in (\ref{eq:thescoretest}) of level $\alpha$ when using the censored likelihood rule (\ref{eq:censlike}) corresponds to the Neyman-Pearson test $\tilde \phi_{NP}$ for (\ref{eq:testingproblem2}) of level $\alpha$.\\
3. If a test $ \phi: \cX^n\to [0,1]$ is constant in arguments varying in $A^c$ for fixed arguments in $A$ for $\mu^{\otimes n}$-a.e.~$(x_1, \ldots, x_n) \in \cX^n$, then its power function is constant on $H_0$ and on $H_1$ in (\ref{eq:testingproblem}). If $\phi$ is redefined to satisfy the requirement everywhere, the values of the power function do not change, and are equal to those of the corresponding test $\tilde \phi$ in 1.~for $\tilde H_0$ and $\tilde H_1$ in (\ref{eq:testingproblem2}).\\ 	
4. For a test $\phi: \cX^n \to [0,1]$, the test $\bar \phi: \tilde \cX \to [0,1]$ defined by
\begin{equation}\label{eq:theavaragedtest}
 \bar \phi(x_1, \ldots, x_n) = \frac{1}{[\mu(A^c)]^{n-j}}\, \int_{A^c}\, \ldots \int_{A^c} \phi(x_1, \ldots, x_j, z_{j+1}, \ldots, z_n)\, d \mu(z_{j+1})\, \ldots \, d\mu( z_{n}),
\end{equation}
where $0 \leq j \leq n$, $x_1, \ldots, x_j \in A$, $x_{j+1}, \ldots, x_n \in A^c,$
and extended by symmetry, does not vary with arguments lying in $A^c$ given a set of fixed arguments in $A$ as described in 1. Its power function is constant on both $H_0$ and on $H_1$, with values corresponding to the power function of $\phi$ for the product density of
\begin{equation}\label{eq:hypotheses}
q_j(x) =  1_A\, p_j(x) + 1_{A^c}\, \frac{P_j(A^c)}{\mu(A^c)}\quad \text{under } H_j,\quad j=0,1.
\end{equation}
\end{lemma}
\begin{proof}[{\sl Proof of Lemma \ref{lem:mainlemtesting}}]
1.~is immediate. For 2., for the censored likelihood score we have that
\[ T(x_1\ldots, x_n) = \sum_{k=1}^n \big[1_A(x_k)\, \log\big(p_1(x_k)/p_0(x_k) \big) + 1_{A^c}(x_k)\, \log\big(P_1(A^c)/P_0(A^c) \big)\big],\]
which corresponds in 1.~to the logarithm of the likelihood ratio statistic for testing $\tilde p_0$ against $\tilde p_1$, the test statistic for the Neyman-Pearson test.\\
Ad 3.~By assumption on $\phi$ we have for $1 \leq j \leq n$ and $x_1, \ldots, x_j \in A$ that
\begin{equation}\label{eq:hilfnotdepend}
 \phi(x_1, \ldots, x_n) = \phi_j(x_1, \ldots, x_j)\quad \text{for } \mu^{\otimes (n-j)}-a.e.~x_{j+1}, \ldots, x_n \in A^c,
\end{equation}
where $\phi_j(x_1, \ldots, x_j)$ is defined by the value on the left side. Now consider the case of the null hypothesis, and suppose that $p$ is such that $\int 1_{A^c} p = P_0(A^c)$. Then
\begin{align}\label{eq:thepowernotdepend}
& \int \ldots \int\,  \phi(x_1, \ldots, x_n)\, \prod_{k=1}^n\, \big(1_A(x_k)\, p_0(x_k) + 1_{A^c}(x_k)\, p(x_k) \big)\, d\mu(x_1)\ldots d\mu(x_n)\notag\\
= & \sum_{j=0}^n\, \binom{n}{j}\, \int_A \ldots \int_A\,\prod_{k=1}^j\,  p_0(x_k)\,\Big( \int_{A^c} \ldots \int_{A^c}\,  \phi(x_1, \ldots, x_n)\notag\\
& \qquad \quad  \prod_{k=j+1}^n\, p(x_k)\, d\mu(x_{j+1})\ldots d\mu(x_n)\Big)\, d\mu(x_{1})\ldots d\mu(x_j)\notag\\
= & \sum_{j=0}^n\, \binom{n}{j}\, [P_0(A^c)]^{n-j}\, \int_A \ldots \int_A\,\phi_j(x_1, \ldots, x_j)\, dP_0(x_{1})\ldots dP_0(x_j),
\end{align}
which does not depend on $p$.

If we define $\tilde \phi: \cX^n \to [0,1]$ by
\[ \tilde \phi(x_1, \ldots, x_j, \bar a, \ldots, \bar a) = \phi_j(x_1, \ldots, x_j),\quad x_1, \ldots, x_j \in A,\]
and extended by symmetry, then a similar computation shows that
\begin{align*}
& \int_{\tilde \cX} \ldots \int_{\tilde \cX}\,  \tilde \phi(x_1, \ldots, x_n)\,  d\tilde P_0(x_1)\ldots d\tilde P_0(x_n)\\
= & \sum_{j=0}^n\, \binom{n}{j}\, [P_0(A^c)]^{n-j}\, \int_A \ldots \int_A\,\phi_j(x_1, \ldots, x_j)\, dP_0(x_{1})\ldots dP_0(x_j),
\end{align*}
which shows the second claim of 3.

Ad 4.~It is clear that $\bar \phi$ does not vary with arguments lying in $A^c$ given a set of fixed arguments in $A$, and its associated $\phi_j$ in (\ref{eq:hilfnotdepend}) is the right side of (\ref{eq:theavaragedtest}). Suppose that $p$ is such that $\int 1_{A^c} p = P_0(A^c)$, then (\ref{eq:thepowernotdepend}) equals
\begin{align*}
& \sum_{j=0}^n\, \binom{n}{j}\, \frac{[P_0(A^c)]^{n-j}}{[\mu(A^c)]^{n-j}}\, \int \ldots \int\,\prod_{k=1}^j\, \big(1_A(x_k)\, p_0(x_k)\big)\,\prod_{k=j+1}^n\,1_{A^c}(x_k)\\
& \qquad \quad   \phi(x_1, \ldots, x_n)\,\, d\mu(x_1)\ldots d\mu(x_n)\\
= & \int \ldots \int\, \bar \phi(x_1, \ldots, x_n)\, \prod_{k=1}^n\, q_0(x_k)\, d\mu(x_1)\ldots d\mu(x_n),
\end{align*}
and similarly under the alternative hypothesis.
\end{proof}
\begin{proof}[{\sl Proof of Theorem \ref{th:scoretesting}}]
Ad 1.~From the property (\ref{eq:locallyproper}) of a localizing weighted scoring rule, the resulting test $\phi$ in (\ref{eq:thescoretest}) for which $T$ in (\ref{eq:scorediff}) is based on $S$ satisfies property 3.~in Lemma \ref{lem:mainlemtesting}. Hence it has constant power function on $H_0$ and on $H_1$ in (\ref{eq:testingproblem}), and the value of the power function over $H_0$ is $\alpha$.

Ad 2.: This follows since the Neyman-Pearson test $\tilde \phi_{NP}$ is uniformly most powerful for (\ref{eq:testingproblem2}), together with Lemma \ref{lem:mainlemtesting}, parts 1., 2.~and 3.
\end{proof}
\begin{proof}[{\sl Proof of Theorem \ref{th:optimalitytest}}]
For the first statement, given any level-$\alpha$-test $\phi:\cX \to [0,1]$ for (\ref{eq:testingproblem2}), the associated test $\bar \phi$ in Lemma \ref{lem:mainlemtesting}, 4., also has level $\alpha$ and further satisfies $\rho(\bar \phi) \geq \rho(\phi)$, since the constant value of the power function of $\bar \phi$ on $H_1$ (see Lemma \ref{lem:mainlemtesting}, 3.) is a particular value of the power function of the test $\phi$.
The statement now follows from Lemma \ref{lem:mainlemtesting}, 1.,2.,3., and optimality of the Neyman-Pearson test for (\ref{eq:testingproblem2}).

The statement for $n=1$ follows from Lemma \ref{lem:mainlemtesting}, 1.~and 3., and Lemma \ref{lem:npcensored} below. This concludes the proof of the theorem.
\end{proof}
\begin{lemma}\label{lem:npcensored}
For any test $\phi: \cX \to [0,1]$ which has level $\alpha$ for (\ref{eq:testingproblem}), e.g.
\begin{equation}\label{eq:level}
 \sup_{p:\ p\, 1_A  = p_0\, 1_A} \Big(\int_A \phi p_0  + \int_{A^c} \phi p\Big) \leq \alpha,
\end{equation}
 the test
\[ \tilde \phi(x) = \phi(x) 1_A(x) + \beta 1_{A^c}(x).\]
where
\[ \beta = \min\Big(\big(\alpha - \int_A \phi p_0\big)/P_0(A^c),1\Big),\]
also has level $\alpha$ and is uniformly more powerful than $\phi$.
\end{lemma}
\begin{proof}[{\sl Proof of Lemma \ref{lem:npcensored}}]
To show that $\tilde \phi$ also has level $\alpha$, observe that any density $p$ under $H_0$ satisfies $\int 1_{A^c} \, p = P_0(A^c)$, hence for  such $p$,
\[ \int \tilde \phi p  = \int_A \phi \, p_0 + \beta\, P_0(A^c) \leq  \alpha\]
by the choice of $\beta$. Thus, $\tilde \phi $ has level $\alpha$.

If $\beta = 1$, then since $\phi(x) \leq 1 $ we have that $\phi(x) \leq \tilde \phi (x)$ for all $x$ and the statement is clear.

If $\beta < 1$, we show that still $\phi(x) \leq \beta$ $\mu$-a.e.~on $A^c$, which implies that $\phi \leq \tilde \phi$ $\mu$-a.e.~and hence $\tilde \phi$ is uniformly more powerful than $\phi$.

Suppose for a contradiction that
\begin{equation}\label{eq:inequality}
\mu\big(\{\phi > \beta\} \cap A^c \big) >0.
\end{equation}
Then the density
\[ p = p_0 1_A + \frac{P_0(A^c)}{\mu\big(\{\phi > \beta\} \cap A^c \big)}\, 1_{\{\phi > \beta\} \cap A^c}\]
satisfies $H_0$, but since $\beta < 1$,
\[ \int \phi p >  \int_A \phi p_0 + \beta P_0(A^c) = \alpha,\]
in contradiction to (\ref{eq:level}).
\end{proof}

\end{document}